\documentclass[11pt]{article}

\usepackage[T1]{fontenc}
\usepackage{amsmath,amssymb,amsthm,mathtools}
\usepackage{thmtools}
\usepackage[margin=1in]{geometry}
\usepackage{float}
\usepackage{tikz}
\usepackage{xcolor}
\usepackage{booktabs,tabularx}
\usepackage{hyperref}

\usepackage[noabbrev,nameinlink,capitalize]{cleveref}
\crefname{appendix}{Appendix}{Appendices}
\Crefname{appendix}{Appendix}{Appendices}

\declaretheorem[numberwithin=section]{theorem}
\declaretheorem[numberlike=theorem]{lemma}
\declaretheorem[numberlike=theorem]{proposition}
\declaretheorem[numberlike=theorem]{corollary}

\declaretheorem[numberlike=theorem,style=definition]{definition}

\usetikzlibrary{arrows.meta}

\hypersetup{
  colorlinks=true,
  linkcolor=blue,
  citecolor=blue,
  urlcolor=blue
}

\definecolor{DarkGreen}{rgb}{0.0, 0.4, 0.0}

\newcommand{\F}{\mathbb{F}}
\newcommand{\E}{\mathbb{E}}
\newcommand{\Prb}{\mathbb{P}}
\newcommand{\cP}{\mathcal{P}}
\newcommand{\OPT}{\operatorname{OPT}}
\newcommand{\cost}{\operatorname{cost}}
\newcommand{\DSN}{\textup{\textsc{DSN}}}

\title{Hardness of Online Directed Steiner Network}
\author{
  \begin{tabular}{c@{\hspace{3em}}c}
    Gary Hoppenworth\thanks{Work done while an intern at Microsoft Research.}
      & Yaowei Long \\
    \normalsize University of Michigan & \normalsize University of Michigan \\
    \normalsize\texttt{garytho@umich.edu} & \normalsize\texttt{yaoweil@umich.edu} \\[3.5ex]
    Sepideh Mahabadi & Jakub Tarnawski \\
    \normalsize Microsoft Research & \normalsize Microsoft Research \\
    \normalsize\texttt{smahabadi@microsoft.com} & \normalsize\texttt{jatarnaw@microsoft.com}
  \end{tabular}
}
\date{}

\begin{document}

\maketitle

\begin{abstract}
In the Directed Steiner Network (DSN) problem we are given a directed graph and a set of demands $(s_i,t_i)$, and asked to find a cheap subgraph connecting each terminal pair. In its online version, the demands arrive online and must be served by buying edges irrevocably.

DSN is a fundamental hard problem in network design, heavily studied in both the offline and the online setting. Offline, it has a superpolylogarithmic hardness of approximation. However, offline hardness says nothing about online algorithms, which are computationally unrestricted. It has been an open question whether uncertainty itself (needing to commit to a solution without knowing future demands) rules out polylogarithmic-competitive online algorithms.

In this work, we show the first such unconditional, information-theoretic hardness. Namely, we give an $\exp\!\bigl(\Omega(\sqrt{\log n})\bigr)$ bound on the competitive ratio, which holds even for randomized algorithms against an oblivious adversary, and on unit-cost DAGs.

Our proof uses a novel connection between online network design and algebraic coding theory. We encode requests using a hidden low-degree polynomial, whose past evaluations reveal nothing about future ones. We then use list-recovery bounds to show that an algorithm cannot make cheaply reusable decisions without knowing those future evaluations.
\end{abstract}

\section{Introduction}

Directed Steiner Network (\DSN) is a basic problem in directed network design.
Given a directed graph with nonnegative edge costs and ordered terminal pairs
$(s_1,t_1),\ldots,(s_k,t_k)$, one seeks a minimum-cost subgraph containing an
$s_i$--$t_i$ path for every $i$.  The problem is also called Directed Steiner
Forest (though a solution need not be a forest)
when every connectivity requirement is one, as it is throughout this
paper.  It contains Directed Steiner Tree as the single-source special case
and is closely connected to directed spanners and reachability preservers.

In the online problem the graph and all edge costs are known from the start,
but the terminal pairs arrive one at a time.  After pair $(s_i,t_i)$ arrives,
the algorithm must irrevocably buy enough edges to connect it before seeing the
next pair.  The benchmark is the
minimum-cost solution that knows the entire demand sequence upfront.  This
model has two key sources of difficulty: finding a good network may
be computationally hard even offline, and committing to a network without
knowing future demands may itself be costly.  Our goal is to study the
second source.

Let $n$ denote the number of vertices and $k$ the number of demand pairs.  For
every fixed $\varepsilon>0$, the best known polynomial-time approximation
ratios for offline \DSN{} with arbitrary costs are
$O(k^{1/2+\varepsilon})$ in terms of $k$
due to Chekuri, Even, Gupta, and Segev~\cite{ChekuriEvenGuptaSegev2011},
and
$O(n^{2/3+\varepsilon})$ in terms of $n$
due to Berman, Bhattacharyya, Makarychev, Raskhodnikova, and Yaroslavtsev~\cite{BermanEtAl2013}.
Unit costs permit a better vertex-dependent bound
of $O(n^{4/7+\varepsilon})$
due to Abboud and Bodwin~\cite{AbboudBodwin2018}.

Online, Chakrabarty, Ene, Krishnaswamy, and Panigrahi give a randomized
$O(k^{1/2+\varepsilon} \mathrm{polylog}(n))$ competitive algorithm for arbitrary
costs~\cite{ChakrabartyEtAl2018}.  In the unit-cost case, Grigorescu, Lin, and
Quanrud obtained an $O(n^{2/3+\varepsilon})$ guarantee~\cite{GrigorescuLinQuanrud2021}; Bodwin and Le recently improved this to
$O(n^{3/5+\varepsilon})$~\cite{BodwinLe2025}.
All these online guarantees are randomized against an oblivious adversary.
No $O(n^\alpha)$-competitive algorithm for general-weight online DSN is known for $\alpha<1$.

The best known guarantees are thus polynomial,
which prompts the natural question:
\emph{is a polylogarithmic approximation/competitive ratio possible?}

In the offline setting, this question was settled negatively already in the 20th century:

\begin{theorem}[Dodis--Khanna~\cite{DodisKhanna1999}]
\label{thm:dodis-khanna}
For every constant $\delta>0$, \DSN{} has no deterministic
polynomial-time
approximation with ratio
$
  2^{(\log n)^{1-\delta}}
$
unless $\mathsf{NP}$ has quasipolynomial-time algorithms.
\end{theorem}

This bound is larger than $\log^c n$ for any $c$.
\emph{Randomized} polynomial-time approximation algorithms with this ratio
can be ruled out
under the slightly stronger assumption that $\mathsf{NP}$ has no \emph{randomized}
quasipolynomial-time algorithms.
Also, their theorem is stated for weighted graphs,
but one can obtain the same statement for unit-cost graphs;
for completeness, we give a brief argument in Appendix~\ref{app:offline-unitization}.

In terms of $k$, Dinur and Manurangsi proved
$k^{1/4-o(1)}$ hardness under ETH, and the same for fixed-parameter
algorithms under Gap-ETH~\cite{DinurManurangsi2018}.  Manurangsi, Rubinstein,
and Schramm obtained an essentially tight $o(\sqrt{k})$ barrier under the
stronger Strongish Planted Clique Hypothesis
~\cite{ManurangsiRubinsteinSchramm2021}.

These results imply that an \emph{efficient} online algorithm cannot
obtain a polylogarithmic competitive ratio,
but they do not hold against algorithms with unlimited computation.
Fundamentally, they do not establish whether \emph{uncertainty} itself -- the core source of difficulty in the online model -- also precludes such algorithms.

Prior to our work,
the best information-theoretic lower bounds for online algorithms
were: $\widetilde{\Omega}(\log^2 n)$ hardness for deterministic algorithms~\cite{AlonEtAl2009},
and $\Omega(\log n)$ hardness for randomized algorithms against an oblivious adversary~\cite[Theorem 2.2.1]{Korman2004}.\footnote{
Both bounds were developed for Online Set Cover, and apply to DSN
via a standard reduction to Directed Steiner Tree:
represent each set $S$ with an edge $(r,S)$ with cost equal to that set's cost, where $r$ is a common root, and add edges $(S,e)$ for every element $e \in S$.
}
See \cref{tab:hardness-summary} for a comparison of the best known lower bounds.

\begin{table}[t]
\centering
\caption{Best known approximation and competitive lower bounds for \DSN{}.
Here $n$ is the number of graph
vertices and $k$ is the number of requests.  ``Unconditional'' means an
information-theoretic lower bound in the online model.  The last
column states whether the result applies to randomized algorithms against an
oblivious adversary.
We note that in \cite{Korman2004}'s $\Omega(k)$ lower bound, the hard instances have $k=\Theta(\log n)$.
}
\label{tab:hardness-summary}
\footnotesize
\renewcommand{\arraystretch}{1.15}
\begin{tabularx}{\textwidth}{@{}
  >{\raggedright\arraybackslash}p{0.2\textwidth}
  >{\raggedright\arraybackslash}p{0.22\textwidth}
  >{\raggedright\arraybackslash}X
  >{\centering\arraybackslash}p{0.11\textwidth}@{}}
\toprule
Reference & Lower bound & Complexity assumptions & Randomized algorithms? \\
\midrule
\cite{DodisKhanna1999}
  & $2^{(\log n)^{1-\delta}}$
  & $\mathsf{NP} \nsubseteq \mathsf{[Randomized] QP}$
  & Yes \\
\cite{DinurManurangsi2018}
  & $k^{1/4-o(1)}$
  & ETH (polynomial time), Gap-ETH (FPT)
  & No \\
\cite{ManurangsiRubinsteinSchramm2021}
  & $\Omega(\sqrt{k})$
  & SPCH, even in FPT time
  & Yes \\
\cite{AlonEtAl2009} 
  & $\Omega(\log^2 n/\log\log n)$
  & \textbf{Unconditional}
  & No \\
\cite{Korman2004} 
  & $\Omega(\log n), \Omega(k)$ 
  & \textbf{Unconditional}
  & Yes \\
\textbf{This work}
  & $\exp(\Omega(\sqrt{\log n}))$
  & \textbf{Unconditional}
  & Yes \\
\bottomrule
\end{tabularx}
\end{table}

\paragraph{Our contribution.}

We prove that no competitive ratio polylogarithmic in $n$ is possible even with unbounded computation.
Our hardness result holds for randomized algorithms against an oblivious adversary.
The hard instances are DAGs with unit edge costs.

\begin{restatable}[Main theorem]{theorem}{mainthm}
\label{thm:main}
There is an absolute constant $c>0$ such that, for every sufficiently large
power of two $R$, there is a fixed unit-cost DAG $J_R$ with the following property.
For every randomized online \DSN{} algorithm $\mathcal A$
there is a demand sequence $\sigma$ on $J_R$, fixed independently of
$\mathcal A$'s random tape $\omega$, for which
\[
  \E_\omega\!\left[\cost_{\mathcal A}(\sigma,\omega)\right]
  \geq
  \exp\!\bigl(c\sqrt{\log |V(J_R)|}\bigr)
  \OPT_{J_R}(\sigma).
\]
\end{restatable}

\paragraph{Our proof technique.}
The proof passes through an intermediate label-covering problem that we call \emph{weighted online Min-Rep}, whose offline variant (Min-Rep) was used in the context of offline hardness for network design problems~\cite{Kortsarz2001}.
The reduction of this problem to DSN is straightforward (see \cref{fig:realization-gadget}).

The central part of the proof uses concepts known in algebraic coding theory, but applies them in a novel way to online competitive analysis.
Our hard instance of
weighted online Min-Rep
is defined using
a vector space $\F_q^d$ over a finite field $\F_q$,
and a hidden low-degree polynomial $P$ over $\F_q^d$.
There is a public sequence $a_1, ..., a_R \in \F_q^d$ of \emph{anchor points},
and over rounds $t=1,...,R$ the values $b_t = P(a_t)$ of the polynomial at these anchor points are revealed.
In round $t$ we consider the set of all \emph{affine lines} $\ell$ through $a_t$ (which partition $\F_q^d \setminus \{a_t\}$).

At a high level,
a solution must purchase \emph{point-labels} $z \in \F_q$ for every point $x \in \F_q^d$,
as well as,
at every round $t$ and every affine line $\ell$ through $a_t$,
\emph{line-polynomial labels} corresponding to low-degree univariate polynomials $p$ on $\ell$ with $p(a_t)=b_t$.
The coverage requirement in round $t$
is that for every point $x \in \F_q^d$
the solution must have purchased a point-label $z \in \F_q$,
and a line-polynomial-label $p$ on the line containing $x$,
that are {compatible}: $p(x)=z$.

An offline optimal solution knows $P$ and can do this cheaply by purchasing $P(x)$ for every $x \in \F_q^d$ and the restriction $P|_\ell$ for every $\ell$.
In particular, the offline solution only makes point-label purchases in round 1, and afterwards only purchases one polynomial-label for each round and line.

On the other hand, an online solution cannot learn anything about future anchor values $b_t$ by observing past ones $b_1,...,b_{t-1}$.
Consider an average line $\ell$.
If the solution has not already overpaid by purchasing many point-labels,
it is now faced with a tradeoff between purchasing a new point-label for many points $x \in \ell$, or purchasing many new polynomial-labels on $\ell$.
Ideally, it would have liked to purchase few new polynomial-labels on $\ell$ that together cover most points on $\ell$ using these points' old point-labels.
However, a Johnson-type~\cite{Johnson1962} list-recovery bound shows that such \emph{prolific} polynomials are very few (much fewer than $q$).
With high probability, no such polynomial is actually allowed for use:
their set is determined before round $t$,
but in round $t$ only polynomials with $p(a_t)=b_t$ are allowed,
and $b_t \in \F_q$ is a fresh uniformly random value.
Thus, the online solution must overpay,
while the offline solution has the very prolific $P|_\ell$ available to it.

We give a more detailed technical overview in
\cref{sec:online-lb-tech-overview}.

\paragraph{Consequences for other problems.}
Our proof technique is not limited to DSN and may be of independent interest.
Of course, \cref{thm:main} also extends to the strictly more general
buy-at-bulk, prize-collecting, or pairwise-spanner variants
of directed network design.
Since all source-to-sink paths in our hard graph have the same length,
the same extends also to online directed pairwise exact-distance preservers~\cite{GrigorescuLinQuanrud2021}.

Our technique can also be applied beyond network design;
we highlight one such consequence here.
Koufogiannakis and Young~\cite{KoufogiannakisYoung2013} study a general online
covering problem in which a solution vector may only increase as arbitrary
upward-closed constraints arrive, and its cost is a monotone submodular
function of that vector.  Our intermediate lower bound
(for weighted online Min-Rep)
yields the following
consequence for their framework.
We give the formal reduction in
Appendix~\ref{app:other-consequences}.

\begin{corollary}[Online submodular-cost covering]
\label{cor:submodular-cost-covering}
For infinitely many $m$, every randomized online submodular-cost covering
algorithm has competitive ratio
\[
  \exp\!\bigl(\Omega(\sqrt{\log m})\bigr),
\]
where $m$ is the number of variables, even against an oblivious
adversary.
\end{corollary}

Finally, it is natural to ask whether our results can extend to undirected network design,
to the single-source (Directed Steiner Tree) setting,
or to fractional solutions.
All three answers are negative.
For online undirected Steiner Forest,
there is a
$O(\log n)$-competitive algorithm
due to
Berman and Coulston~\cite{BermanC97};
and for Directed Steiner Tree,
Chakrabarty, Ene, Krishnaswamy, and Panigrahi~\cite[Corollary 22]{ChakrabartyEtAl2018}
gave
a (quasipolynomial-time) online algorithm
with $\mathrm{polylog}(n)$ competitive ratio.\footnote{
    In fact it's a well-known open problem to give such an algorithm with polynomial runtime, even offline;
    one can get $\widetilde{O}(k^\epsilon)$ for any $\epsilon>0$ in polynomial time.
}
As for fractional solutions,
the online fractional covering framework of
Alon, Awerbuch, Azar, Buchbinder, and Naor~\cite{AlonEtAl2006}
and Buchbinder and Naor~\cite{BuchbinderNaor2009}
allows one to maintain such a non-decreasing fractional solution online
with competitive ratio $O(\log n)$.\footnote{
  Indeed, this is consistent with the structure of the known online algorithms for
  \DSN{}~\cite{ChakrabartyEtAl2018,GrigorescuLinQuanrud2021,BodwinLe2025},
  all of which solve such a relaxation online at a polylogarithmic loss
  and incur their polynomial factors in rounding.
}
Thus it is only the combination
of directedness, multiple source-sink pairs, and integrality that makes online \DSN{} superpolylogarithmically hard.

\paragraph{Outline.}

In \cref{sec:related-work} we discuss related work, and in particular
compare our approach with the offline hardness proof of Dodis and Khanna.
\cref{sec:online-lb-tech-overview} gives a high-level overview and explanation of the proof.
In \cref{sec:preliminaries} we formally define the considered problems,
some used notions,
and develop the mostly coding-theoretic tools
(for which we give simple, self-contained proofs).
In \cref{sec:hidden-polynomial}
we define our hard instance (distribution),
and in \cref{sec:distributional-lower-bound}
we show that any randomized online algorithm
incurs a high expected cost on that distribution.

\subsection{Related work}
\label{sec:related-work}

\paragraph{Comparison with the offline hardness of Dodis and Khanna.}
Our proof and the offline hardness proof of Dodis and
Khanna~\cite{DodisKhanna1999} both pass via a
Min-Rep-type label covering problem, followed by a reduction to \DSN{}.
Our reduction has extra ingredients to obtain unit-cost graphs,
but both reductions are straightforward.
The main difference lies in the hardness construction for the
label-covering problem.
Dodis and Khanna call the intermediate problem \emph{Symmetric Label Cover}
and invoke its known hardness via Label Cover~\cite{KhannaSudanTrevisan1997}.
Khanna, Sudan, and Trevisan~\cite{KhannaSudanTrevisan1997} in turn attribute
Total Label Cover hardness to the low-error multiprover proof systems of Raz and
Safra~\cite{RazSafra1997} and Arora and Sudan~\cite{AroraSudan1997}.
These PCP-based hardness arguments separate satisfiable instances, where
one label per variable suffices, from unsatisfiable instances, where every
cover must use many labels.  The distinction
between cheap and expensive instances is thus computational: the cheap
solution, when it exists, is determined by the input but hard to find.

Our lower bound for weighted online Min-Rep has to be of a
different nature.  An online algorithm with unlimited computation can solve
every prefix it has seen optimally, so the cheap solution cannot be hidden
computationally; it must be hidden in the future of the sequence.
Accordingly, every prefix of our constraint sequence has many cheap
solutions, one per polynomial consistent with the revealed anchor values.
The difficulty is choosing point labels that can be reused cheaply in later
rounds.  The hidden polynomial $P$ provides the analogue of completeness,
but our lower-bound argument differs from PCP soundness in two ways.
First, PCP soundness rules out cheap solutions on unsatisfiable instances,
whereas our bound concerns reuse of the algorithm's old point labels.
On a fixed line, list recovery (\cref{lem:johnson}) bounds the number of
polynomials that agree with the old point-label lists at many points whose
lists are short.  This list of prolific polynomials is determined before
round $t$, and with high probability the fresh anchor value $b_t$ makes
all of them unavailable as line labels in that round
(\cref{lem:prolific-list}).  Meanwhile, the hidden polynomial still
certifies a cheap offline solution using its own point labels $P(x)$
together with the line restrictions $P|_\ell$.
Second, we use only an elementary Johnson-type bound, whereas the low-error PCP constructions
must encode an NP-hard problem.  We are free to choose the algebraic
instance because no computational hardness is needed.  The low-degree
tests cited above also use polynomials, lines, and list-size bounds, but
to certify that a table is close to a polynomial; our approach does not
involve tests or certification.

\paragraph{Coding-theoretic online lower bounds.}
Böckenhauer, Hromkovič, Komm, Krug, Smula, and Sprock~\cite{BockenhauerEtAl2014} use the fact that a bounded
collection of advice strings cannot closely match every possible input string
to prove advice lower bounds for Online Set Cover and a relaxed Online Maximum
Clique objective.  Angelopoulos and Kamali~\cite{AngelopoulosKamali2023}
similarly exploit the ambiguity left by erroneous advice for online search,
bidding, and fractional knapsack.  Both results concern bounded or imperfect
advice, rather than unrestricted randomized algorithms.

\paragraph{Hidden functions under online erasures.}
A closer analogue of our hidden input appears in property testing with online
erasures.  Kalemaj, Raskhodnikova, and Varma~\cite{KalemajRaskhodnikovaVarma2023} and Ben-Eliezer, Kelman, Meir, and Raskhodnikova~\cite{BenEliezerEtAl2024} hide random linear or low-degree functions by
adaptively erasing values inferable from the tester's previous queries.  These
are query-complexity lower bounds against an adaptive erasure adversary, rather
than competitive lower bounds against an obliviously generated input sequence.

\paragraph{Polynomial bounds when the graph also arrives online.}
A \emph{reachability preserver} for a set of terminal pairs is a subgraph
containing a path for each pair, that is, a feasible solution to unit-cost
\DSN{}.  Polynomial bounds are known in a stronger, nonstandard online model:
Bodwin, Hoppenworth, and Trabelsi~\cite{BodwinHoppenworthTrabelsi2023}
showed that
if the adversary may add edges to the input graph before each request, it can
force every online reachability preserver for $p$ requests on $n$ vertices to
use $\Omega((np)^{2/3}+n)$ edges.

\section{Technical overview}
\label{sec:online-lb-tech-overview}

We give a more detailed but still high-level and intuitive explanation of our online hardness proof.

\paragraph{The intermediate problem and its network realization.}
The proof passes through an intermediate covering problem that we call
\emph{weighted online Min-Rep} (\cref{def:weighted-online-min-rep}).
In this problem
there is a fixed bipartite set of left and
right variables.  Each variable has a set of labels; buying label $\lambda$
of variable $w$ has a positive integer cost $c(w,\lambda)$.  A possible
constraint between a left variable $u$ and a right variable $v$ specifies a relation
$\mathcal R_{uv}$ of compatible label pairs.
Constraints $(u,v)$ arrive online.
After each arrival, the algorithm must irrevocably buy labels
so that the constraint is satisfied,
which means that the labels bought for $u$ and $v$ must contain a
compatible pair.
A bought label may be reused for every constraint incident
to its variable.
The objective to minimize is the total cost of the labels
bought.

This problem reduces directly to online unit-cost \DSN{}.
See \Cref{fig:realization-gadget}.
Replace every
label by a directed chain whose length is its cost.  For every possible
constraint $(u,v)$, add private source and sink terminals and connectors so
that every source--sink path must traverse one label chain of $u$, a connector
for a compatible pair, and one label chain of $v$.
Conversely, any compatible pair gives
such a path.
See \cref{lem:realization} for more details.

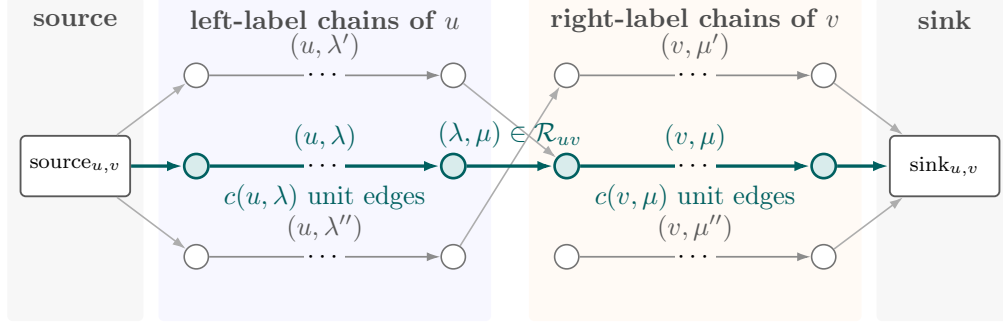
\begin{figure}[t]
  \centering
  \begin{tikzpicture}[
      x=1cm,
      y=1cm,
      >={Latex[length=2.1mm,width=1.4mm]},
      every node/.style={font=\small},
      terminal/.style={
        rectangle,
        rounded corners=2pt,
        draw=black!65,
        fill=white,
        line width=0.7pt,
        minimum width=1.45cm,
        minimum height=8mm,
        inner sep=2pt,
        font=\scriptsize
      },
      endpoint/.style={
        circle,
        draw=black!55,
        fill=white,
        line width=0.6pt,
        minimum size=3.2mm,
        inner sep=0pt
      },
      selected endpoint/.style={
        endpoint,
        draw=teal!75!black,
        fill=teal!15,
        line width=1pt
      },
      chain/.style={-{Latex[length=1.8mm,width=1.2mm]},draw=black!45,line width=0.65pt},
      connector/.style={-{Latex[length=1.8mm,width=1.2mm]},draw=black!32,line width=0.6pt},
      selected/.style={-{Latex[length=2.1mm,width=1.4mm]},draw=teal!75!black,line width=1.25pt},
      layer title/.style={font=\small\bfseries,text=black!70,align=center}
    ]
    \fill[black!3,rounded corners=2pt] (-0.75,-2.05) rectangle (1.05,2.25);
    \fill[blue!3,rounded corners=2pt] (1.25,-2.05) rectangle (5.65,2.25);
    \fill[orange!4,rounded corners=2pt] (6.15,-2.05) rectangle (10.55,2.25);
    \fill[black!3,rounded corners=2pt] (10.75,-2.05) rectangle (12.55,2.25);

    \node[layer title] at (0.15,1.93) {source};
    \node[layer title] at (3.45,1.93) {left-label chains of $u$};
    \node[layer title] at (8.35,1.93) {right-label chains of $v$};
    \node[layer title] at (11.65,1.93) {sink};

    \node[terminal] (source) at (0.15,0)
      {$\operatorname{source}_{u,v}$};
    \node[terminal] (sink) at (11.65,0)
      {$\operatorname{sink}_{u,v}$};

    \node[endpoint] (ul-top-minus) at (1.75,1.2) {};
    \node[endpoint] (ul-top-plus) at (5.15,1.2) {};
    \node[selected endpoint] (ul-minus) at (1.75,0) {};
    \node[selected endpoint] (ul-plus) at (5.15,0) {};
    \node[endpoint] (ul-bottom-minus) at (1.75,-1.2) {};
    \node[endpoint] (ul-bottom-plus) at (5.15,-1.2) {};

    \node[endpoint] (vr-top-minus) at (6.65,1.2) {};
    \node[endpoint] (vr-top-plus) at (10.05,1.2) {};
    \node[selected endpoint] (vr-minus) at (6.65,0) {};
    \node[selected endpoint] (vr-plus) at (10.05,0) {};
    \node[endpoint] (vr-bottom-minus) at (6.65,-1.2) {};
    \node[endpoint] (vr-bottom-plus) at (10.05,-1.2) {};

    \draw[chain] (ul-top-minus) -- node[midway,fill=blue!3,inner sep=1pt] {$\cdots$} (ul-top-plus);
    \draw[selected] (ul-minus) -- node[midway,fill=blue!3,inner sep=1pt] {$\cdots$} (ul-plus);
    \draw[chain] (ul-bottom-minus) -- node[midway,fill=blue!3,inner sep=1pt] {$\cdots$} (ul-bottom-plus);
    \draw[chain] (vr-top-minus) -- node[midway,fill=orange!4,inner sep=1pt] {$\cdots$} (vr-top-plus);
    \draw[selected] (vr-minus) -- node[midway,fill=orange!4,inner sep=1pt] {$\cdots$} (vr-plus);
    \draw[chain] (vr-bottom-minus) -- node[midway,fill=orange!4,inner sep=1pt] {$\cdots$} (vr-bottom-plus);

    \draw[connector] (source) -- (ul-top-minus);
    \draw[selected] (source) -- (ul-minus);
    \draw[connector] (source) -- (ul-bottom-minus);
    \draw[connector] (ul-top-plus) -- (vr-minus);
    \draw[connector] (ul-bottom-plus) -- (vr-top-minus);
    \draw[selected] (ul-plus) -- (vr-minus);
    \draw[connector] (vr-top-plus) -- (sink);
    \draw[selected] (vr-plus) -- (sink);
    \draw[connector] (vr-bottom-plus) -- (sink);

    \node[above=2pt,text=black!65] at (3.45,1.2) {$(u,\lambda')$};
    \node[above=2pt,text=teal!75!black] at (3.45,0) {$(u,\lambda)$};
    \node[above=2pt,text=black!65] at (3.45,-1.2) {$(u,\lambda'')$};
    \node[above=2pt,text=black!65] at (8.35,1.2) {$(v,\mu')$};
    \node[above=2pt,text=teal!75!black] at (8.35,0) {$(v,\mu)$};
    \node[above=2pt,text=black!65] at (8.35,-1.2) {$(v,\mu'')$};

    \node[below=3pt,text=teal!75!black] at (3.45,0) {$c(u,\lambda)\text{ unit edges}$};
    \node[below=3pt,text=teal!75!black] at (8.35,0) {$c(v,\mu)\text{ unit edges}$};
    \node[above=3pt,text=teal!75!black] at (5.9,0)
      {$(\lambda,\mu)\in\mathcal R_{uv}$};
  \end{tikzpicture}
  \caption{Realizing one weighted online Min-Rep constraint $(u,v)$ as a
  \DSN{} demand, done as part of the reduction of \cref{lem:realization}.  Each arrow containing $\cdots$ is a unit-edge label
  chain.  The highlighted path selects a compatible pair
  $(\lambda,\mu)\in\mathcal R_{uv}$.  Label chains are shared across
  constraints; terminals and connectors (the edges not on chains) are private to $(u,v)$.}
  \label{fig:realization-gadget}
\end{figure}

\paragraph{Independent evaluations from a hidden polynomial.}
Fix a large power of two $R$, set $d=\log_2 R$ and $q=R^{10}$, and let $\cP$
be the space of $d$-variate polynomials over $\F_q$ of total degree at most
$d$.
Since $\dim\cP\geq R$, there exist public \emph{anchor points}
$a_1,\ldots,a_R\in\F_q^d$ whose evaluation maps are linearly independent
(this is related to the coding-theoretic notion of \emph{information set}; see \cref{lem:independent-evaluations}).
We draw a uniformly random $P\in\cP$.
The values
\[
  b_t=P(a_t),\qquad t=1,\ldots,R,
\]
are then independent and uniform in $\F_q$.  Round $t$ reveals $b_t$.
These values will determine which constraints arrive in each round.
Knowing $b_1, ..., b_{t-1}$ gives no information about $b_t$,
which is the source of the hard uncertainty;
on the other hand, the offline optimum can know $P$.

The logarithmic degree $d$ gives enough linearly-independent evaluations while
keeping the final graph size quasipolynomial in $R$.

\paragraph{Affine lines through an anchor.}
For a point $a\in\F_q^d$ and a nonzero direction
$v\in\F_q^d$, the affine line through $a$ in direction $v$ is
\[
  \ell=\{a+sv:s\in\F_q\}.
\]
It has $q$ points and the above formula parametrizes it by $s\in\F_q$.
  We call its other $q-1$ points the \emph{nonanchor points}.  Rescaling $v$
  gives another parametrization of the same line.
  
The lines through a fixed
  anchor $a_t$ intersect only at $a_t$ and partition
$\F_q^d\setminus\{a_t\}$.  In particular, there are
\[H=(q^d-1)/(q-1)\] such lines.

Restricting a multivariate polynomial $P$ of
total degree at most $d$ to $\ell$ gives the univariate polynomial
$s\mapsto P(a_t+sv)$, again of degree at most $d$, which is denoted $P|_\ell$.
Conversely, when we speak
of a degree-at-most-$d$ \emph{line polynomial} on $\ell$, we mean a function that becomes
such a univariate polynomial under a fixed parametrization sending $0$ to the
anchor.

\paragraph{The hidden-polynomial Min-Rep instance.}
We now define our hard weighted online Min-Rep instance.
Write $N=q^d$.  For every point
$x\in\F_q^d$, take $x$ itself as a right variable.  Its labels (\emph{point labels}) are the
possible values $z\in\F_q$, each of cost $c_p=100R$; buying $z$ records that
value at $x$.
For every round $t$, possible anchor value $b\in\F_q$, and
affine line $\ell$ through $a_t$, take the triple $(t,b,\ell)$ as a left
variable.  Its labels are the degree-at-most-$d$ line polynomials $g$ on $\ell$
satisfying $g(a_t)=b$, each of cost $c_r=100q$.  For every $x\in\ell$, the
possible constraint between $(t,b,\ell)$ and $x$ declares $(g,z)$ compatible
exactly when $g(x)=z$.  \Cref{fig:hidden-polynomial-gadget} shows this
Min-Rep constraint.  
This defines the public (known in advance) part of the instance.
It remains to define the online constraint sequence.

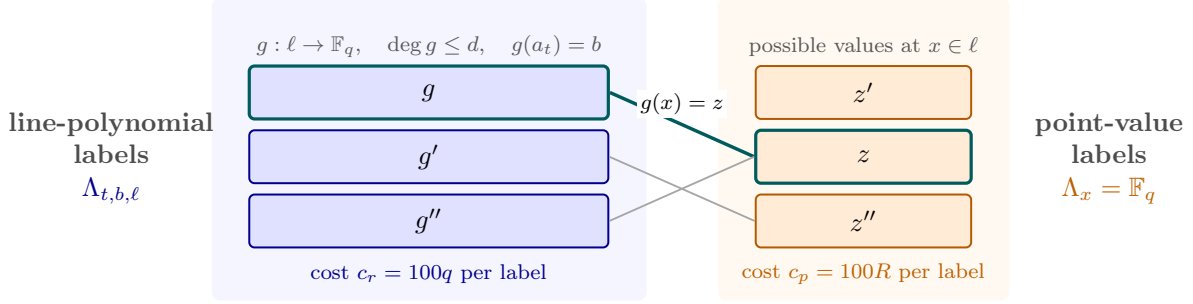
\begin{figure}[t]
  \centering
  \begin{tikzpicture}[
      x=1cm,
      y=1cm,
      every node/.style={font=\small},
      polynomial label/.style={
        rounded corners=2pt,
        draw=blue!55!black,
        fill=blue!12,
        line width=0.75pt,
        minimum width=4.75cm,
        minimum height=7mm,
        align=center,
        inner sep=2pt
      },
      value label/.style={
        rounded corners=2pt,
        draw=orange!70!black,
        fill=orange!16,
        line width=0.75pt,
        minimum width=2.85cm,
        minimum height=7mm,
        align=center,
        inner sep=2pt
      },
      compatibility/.style={draw=black!35,line width=0.7pt},
      selected compatibility/.style={draw=teal!72!black,line width=1.3pt},
      layer title/.style={font=\small\bfseries,text=black!70,align=center}
    ]
    \fill[blue!4,rounded corners=3pt] (1.55,-1.9) rectangle (7.3,2.1);
    \fill[orange!5,rounded corners=3pt] (8.25,-1.9) rectangle (12.1,2.1);

    \node[layer title] at (0.22,0.25) {line-polynomial\\labels};
    \node[font=\small\bfseries,text=blue!55!black] at (0.22,-0.45)
      {$\Lambda_{t,b,\ell}$};
    \node[layer title] at (13.42,0.25) {point-value\\labels};
    \node[font=\small\bfseries,text=orange!75!black] at (13.42,-0.45)
      {$\Lambda_x=\F_q$};

    \node[font=\scriptsize,text=black!65] at (4.42,1.46)
      {$g:\ell\to\F_q,\quad \deg g\leq d,\quad g(a_t)=b$};
    \node[font=\scriptsize,text=black!65] at (10.17,1.46)
      {possible values at $x\in\ell$};

    \node[polynomial label,draw=teal!72!black,line width=1.2pt]
      (g) at (4.42,0.85) {$g$};
    \node[polynomial label] (g-prime) at (4.42,0) {$g'$};
    \node[polynomial label] (g-double) at (4.42,-0.85) {$g''$};

    \node[value label] (z-prime) at (10.17,0.85) {$z'$};
    \node[value label,draw=teal!72!black,line width=1.2pt]
      (z) at (10.17,0) {$z$};
    \node[value label] (z-double) at (10.17,-0.85) {$z''$};

    \node[font=\scriptsize,text=blue!55!black] at (4.42,-1.55)
      {cost $c_r=100q$ per label};
    \node[font=\scriptsize,text=orange!75!black] at (10.17,-1.55)
      {cost $c_p=100R$ per label};

    \draw[selected compatibility] (g.east) --
      node[midway,above=2pt,fill=white,inner sep=1.2pt,font=\scriptsize]
      {$g(x)=z$}
      (z.west);
    \draw[compatibility] (g-prime.east) -- (z-double.west);
    \draw[compatibility] (g-double.east) -- (z.west);
  \end{tikzpicture}
  \caption{One weighted online Min-Rep constraint between $(t,b,\ell)$ and
  $x$.  A middle edge joins a compatible pair.  The highlighted pair $(g,z)$
  covers the constraint because $g(x)=z$.}
  \label{fig:hidden-polynomial-gadget}
\end{figure}

The constraints arrive online in rounds $t = 1, ..., R$.
In each round we reveal a block of constraints,
and inspect the purchases made by an online algorithm
by the end of the round.

Recall that points $a_1,...,a_R \in \F_q^d$ are public,
and that we sampled a hidden polynomial $P$,
which defines values $b_t=P(a_t)$.
The block revealed in round $t$ contains those of the above constraints where $b=b_t$,
one for every line $\ell$ through $a_t$ and every $x\in\ell$.
(The
Min-Rep instance contains variables and constraints for every candidate $b$,
but only the block selected by the newly revealed $b=b_t$ arrives.)

For every realized sequence, its generating polynomial $P$ certifies a cheap
offline solution: buy the point label $P(x)$ once at every point and the line
label $P|_\ell$ on every requested line.
In particular, this solution  buys point labels only in round 1.
The point labels cost $100RN$, and
the line labels cost $100qRH=\Theta(RN)$ because $Hq=\Theta(N)$.  Their total
cost is less than $250RN$.  We use the slightly larger benchmark
\[
  M=300RN
\]
to have some slack for the reduction to \DSN{}.
Therefore, our goal is to establish a Min-Rep lower bound showing that every online
strategy has expected label cost $\Omega(RM)$
(that is, $\Omega(M)$ on average in every round).

\paragraph{Prolific polynomials.}
Just before round $t$, call point labels bought in earlier rounds \emph{old}.
They form an old-value list at each point.  Call a point \emph{heavy} if its old-value list
has more than $16R$ values, and call a line through $a_t$ \emph{crowded} if at least half its points are heavy.  
As long as the accumulated cost is below $RM=300R^2N$, the cost $100R$ per
point label permits fewer than $3RN$ point labels in total.
Intuitively, this means that we can focus our attention only on non-crowded lines and non-heavy points.

So fix a round $t$ and a line $\ell$.
Recall that the algorithm has to buy some line polynomials on this line,
and some new point labels, such that every point on $\ell$ has a label (old or new) covered by one of the line polynomials.
If the algorithm buys $\Omega(q)$ new point-labels every time, it will overpay.
So most points should be served using old labels.
Ideally, the algorithm now wants to purchase few line polynomials to serve these points; thus, it wants to buy \emph{prolific} polynomials, each of which will serve many points
(i.e., take one of the values in $x$'s old-list, for many $x$).
However, nonheavy points have short old-lists.
A crucial step now is a Johnson-type~\cite{Johnson1962}
list-recovery bound
that shows that there are few low-degree prolific polynomials.

More precisely,
on each line we call a degree-$d$ polynomial \emph{prolific} if it takes an old listed
value on at least $\rho q$ points, where $\rho=1/(256R)$.  The Johnson-type
bound (\cref{lem:johnson}) says that there are only $O(R^3)$ prolific polynomials.
Their values at $a_t$ form a candidate set determined
entirely before $b_t$ is revealed.
Since $b_t$ remains uniform
given the past,
and $q=R^{10}$,
this candidate set contains $b_t$
only with small probability.
Recall that the algorithm is only allowed to use line polynomials that have $p(a_t)=b_t$;
most likely, no prolific polynomial is allowed
(\cref{lem:prolific-list}).

\paragraph{The resulting line-label/point-label tradeoff.}
Fix a noncrowded line and suppose its candidate set misses $b_t$.  Every line
polynomial available in round $t$ is then nonprolific.  If the algorithm buys
$u$ such polynomials, together they can support old labels on fewer than
$u\rho q$ of the at least $q/2$ nonheavy points.  Thus at least
$q(1/2-\rho u)$ points require new point labels.  Since a line label
costs $100q$ and a point label costs $100R$, the resulting cost is at least
\[
  100q u + 100R q (1/2 - \rho u) =
  100q\bigl(u+R/2 - u/256\bigr)=\Omega(Rq)
\]
(\cref{lem:line-tradeoff}).

Distinct lines through $a_t$ have disjoint nonanchor points, so the line-label
and new point-label costs charged to different lines in one round are
disjoint.  A constant fraction of the $H$ lines are noncrowded, and each misses
the candidate set with high  probability.  The
conditional expected cost of an unfinished round is therefore
$\Omega(HRq)=\Omega(RN)=\Omega(M)$ (\cref{lem:round-increment}).  Summing these increments over $R$
rounds gives the desired bound.  If the final cost reaches $RM$
with probability at least $1/2$, its expectation is already $\Omega(RM)$.
Otherwise, every round starts with cost below $RM$ with probability greater than
$1/2$, so summing the conditional $\Omega(M)$ increment over all $R$ rounds
again gives expected cost $\Omega(RM)$.

\paragraph{Graph size.}
The number
$n_R$ of vertices in the final DSN graph satisfies
$\log n_R=O((\log R)^2)$.  The competitive gap $\Omega(R)$ is therefore
$\exp(\Omega(\sqrt{\log n_R}))$.

\paragraph{Discussion and open problems.}
The gap between $\exp(\Omega(\sqrt{\log n}))$ and the best upper bounds,
$O(n^{3/5+\varepsilon})$~\cite{BodwinLe2025} for unit costs and
$\widetilde O(k^{1/2+\varepsilon})$~\cite{ChakrabartyEtAl2018} in general, remains
large.  The exponent $\sqrt{\log n}$ is a limitation of our construction and
not of its analysis.\footnote{
Indeed, it is easy to check that an online algorithm
that buys all $q$ point labels of
a point when the point is first requested, and one arbitrary legal line label
for every line in every round,
gets competitive ratio $O(q)=R^{O(1)}$ on $J_R$
(it pays at most $100RNq+100qRH+3RHq=O(qRN)$, while the offline optimum has to buy at least $N$ point labels
and therefore pays at least $100RN$).
Since $\log n_R=\Theta(\log^2 R)$, the
competitive ratio is $\exp(\Theta(\sqrt{\log n_R}))$, and only a
different construction can improve \Cref{thm:main}.
}
The main question left open by our work is to narrow this gap:
Is there an information-theoretic lower bound
of $n^{\Omega(1)}$ for online \DSN{}?  Perhaps against deterministic algorithms only?
Somewhat more concretely,
our construction can be viewed as Reed--Solomon
codes on the lines of a Reed--Muller code; other locally structured codes may
do better.

\section{Preliminaries}
\label{sec:preliminaries}

The following parameters remain fixed throughout the paper. Let $R$ be a sufficiently large power of two.
Set
\[
  d=\log_2R,\qquad q=R^{10},\qquad N=q^d.
\]
Since $R$ is a power of two, $q$ is an even prime power, so the field $\F_q$ exists.

\paragraph{Online DSN.} We first formally define the online (unit-cost) DSN problem.

\begin{definition}[Online unit-cost \DSN]
An instance consists of a fixed directed graph $J=(V,A)$ in which every edge
has cost $1$.  Ordered pairs $(s_1,t_1),(s_2,t_2),\ldots$ arrive one at a time.
After each arrival, the algorithm may irrevocably buy edges and must ensure
that every revealed $t_i$ is reachable from $s_i$ in the purchased subgraph.
The offline optimum knows the entire demand sequence in advance.
\end{definition}

A randomized algorithm is $\gamma$-competitive against an oblivious adversary
if, for every fixed request sequence $\sigma$,
\[
  \E_\omega[\cost_{\mathcal A}(\sigma,\omega)]
  \leq \gamma\,\OPT(\sigma),
\]
where the random tape $\omega$ is independent of $\sigma$.  Our lower bound
places no restriction on the running time or computation used to choose
purchases.

\paragraph{Weighted online Min-Rep and reduction to DSN.}

We first separate the online combinatorial lower bound from its graph
realization.  Classical Min-Rep~\cite{Kortsarz2001} selects representative vertices from groups
so that every superedge has adjacent selected representatives.  
In addition, we allow each
representative to have a positive integer cost, and we reveal superedges online;
we call the resulting problem \emph{weighted online Min-Rep}.  We use the
following bipartite form.

\begin{definition}[Weighted online Min-Rep]
\label{def:weighted-online-min-rep}
An instance consists of disjoint sets $U$ and $V$ of left and right
variables.  Every variable $w\in U\cup V$ has a label set $\Lambda_w$,
and buying $\lambda\in\Lambda_w$ costs the positive integer $c(w,\lambda)$.\footnote{Labels are independent between variables, even when label sets for different variables are not disjoint. Buying a label $\lambda$ for $w$ does not buy $\lambda$ for any other variable.}
The set of possible constraints is indexed by a set
$E\subseteq U\times V$.  For each $(u,v)\in E$ there is one compatibility
relation
\[
  \varnothing\neq\mathcal R_{uv}\subseteq\Lambda_u\times\Lambda_v.
\]
In particular, there is at most one possible constraint on each ordered
left--right variable pair $(u,v)$.  

An online sequence will reveal constraints. When a constraint $(u,v)$ is revealed, the labels
already purchased, together with any labels bought in response, must contain
some $(\lambda,\mu)\in\mathcal R_{uv}$.  Label purchases are irrevocable, and
the cost is the total cost of the distinct labels purchased.
For a request sequence $\sigma$, write $\OPT_{\mathrm{MinRep}}(\sigma)$ for the
minimum total label cost of an offline solution.
\end{definition}

The reduction below is a variant of the standard Label-Cover-to-\DSN{}
gadget of Dodis and
Khanna~\cite{DodisKhanna1999}; see also
Chitnis, Feldmann, and Manurangsi~\cite{ChitnisFeldmannManurangsi2018}.
See \cref{sec:related-work} for a comparison.
In our formulation, each integer label cost is realized by a private
unit-cost chain of that length, each constraint has private terminals and
connectors, and we track how edge sets bought for a prefix of the demands
translate into label sets, so that the reduction applies to online algorithms.
See \Cref{fig:realization-gadget} for an illustration.
The construction is straightforward; we give a proof for completeness.

\begin{lemma}[Min-Rep-to-\DSN{} realization]
\label{lem:realization}
Every weighted online Min-Rep instance has a fixed unit-cost DAG $J$
with a demand
$(\operatorname{source}_{u,v},\operatorname{sink}_{u,v})$ for every
$(u,v)\in E$ with the following properties.
The construction represents each label by a dedicated directed path, called
its label chain.  For an edge set $F$, define
\[
  \Lambda(F):=\{(w,\lambda):
  \text{the entire label chain of }(w,\lambda)\text{ lies in }F\}.
\]
The map $F\mapsto\Lambda(F)$ is monotone, and for every requested constraint
set $S\subseteq E$:
\begin{enumerate}
  \item If $F$ connects the demands in $S$, then $\Lambda(F)$ satisfies $S$
  and
  \begin{equation}
    |F|\geq
    \sum_{(w,\lambda)\in\Lambda(F)}c(w,\lambda)+3|S|.
    \label{eq:realization-lower}
  \end{equation}
  \item Conversely, any label set $L$ satisfying $S$ gives an edge set that
  connects the demands in $S$ and has cost exactly
  \[
    \sum_{(w,\lambda)\in L}c(w,\lambda)+3|S|.
  \]
\end{enumerate}
Consequently, every online \DSN{} algorithm on $J$ induces, prefix by prefix,
a weighted online Min-Rep strategy whose label cost is at most its edge cost.
\end{lemma}

\begin{proof}
For each variable--label pair $(v,\lambda)$, create a directed path
\[
  v_\lambda^0\longrightarrow v_\lambda^1\longrightarrow\cdots
  \longrightarrow v_\lambda^{c(v,\lambda)}
\]
of length $c(v,\lambda)$.  All of its vertices and edges are private to that
label.  Write $v_\lambda^-=v_\lambda^0$ and
$v_\lambda^+=v_\lambda^{c(v,\lambda)}$.

For each $(u,v)\in E$, create private terminals
$\operatorname{source}_{u,v},\operatorname{sink}_{u,v}$.  Add a source
connector for every
$\lambda\in\Lambda_u$, a middle connector for every compatible pair
$(\lambda,\mu)\in\mathcal R_{uv}$, and a sink connector for every
$\mu\in\Lambda_v$:
\begin{equation}
  \operatorname{source}_{u,v}\longrightarrow u_\lambda^-,\qquad
  u_\lambda^+\longrightarrow v_\mu^-,\qquad
  v_\mu^+\longrightarrow\operatorname{sink}_{u,v}
  \label{eq:connectors}
\end{equation}
These connectors and the label-chain edges are all the edges, and each has
cost $1$.

Every $\operatorname{source}_{u,v}$--$\operatorname{sink}_{u,v}$ path has a
forced form: a source connector, an entire left-label chain, a middle
connector, an entire right-label chain, and a sink connector.  Indeed, the
internal vertices of a label chain have no external incidences, and every
chain has at least one edge because label costs are positive.  The source and
sink connectors force the two labels to belong to $u$ and $v$, respectively.
Since there is at most one constraint on the ordered pair $(u,v)$ by
\cref{def:weighted-online-min-rep}, the middle connector for this
pair joins labels in $\mathcal R_{uv}$.  Thus $\Lambda(F)$ satisfies $S$.
This also shows that a partially purchased chain cannot help connect a demand.

Distinct label chains are edge-disjoint.  Source and sink connectors are
private to their constraint, and ordered-pair uniqueness makes every middle
connector private as well.  A connected demand therefore accounts for three
distinct connectors, proving~\eqref{eq:realization-lower}.  Adding edges can
only complete more chains, so $F\mapsto\Lambda(F)$ is monotone.

Conversely, given a label set $L$ satisfying $S$, buy the label chain of every
label in $L$.  For each $(u,v)\in S$, choose
$(\lambda,\mu)\in\mathcal R_{uv}$ with
$(u,\lambda),(v,\mu)\in L$ and buy the corresponding three connectors.
Shared label chains are bought only once, while the connectors are private to
their constraints, so the resulting cost is exactly the claimed sum.  Ordering
the vertices as sources, left chains, right chains, and sinks shows that the
graph is acyclic.

Finally, let $F_i$ be the edges bought by an online \DSN{} algorithm after a
request prefix $S_i$.  The sets $F_i$, and hence $\Lambda(F_i)$, grow
monotonically.  The first property shows that $\Lambda(F_i)$ satisfies $S_i$
and has label cost at most $|F_i|$.  Buying each label when it first appears
in $\Lambda(F_i)$ therefore defines the claimed weighted online Min-Rep
strategy.
\end{proof}

Now we turn to developing the linear-algebraic tools
that will be needed for the proof.
First we record some basic facts about affine lines:

\begin{lemma}[Affine-line geometry]
\label{lem:pencil-geometry}
For every $a\in\F_q^d$, the affine lines through $a$ have $q$ points each,
and their nonanchor point sets partition $\F_q^d\setminus\{a\}$.  Their
number is
\[
  H:=\frac{N-1}{q-1}=1+q+\cdots+q^{d-1}
\]
and we have
\[
  \frac Nq\leq H\leq2q^{d-1},
  \qquad
  \frac{Hq}{N}<\frac{q}{q-1}.
\]
\end{lemma}

\begin{proof}
Every line through $a$ has the form
$\{a+sv:s\in\F_q\}$ for a nonzero direction $v$.  It has $q$ points because
$s\mapsto a+sv$ is injective.  Every $x\neq a$ lies on a unique such line,
namely the line in direction $x-a$.  Thus the nonanchor point sets partition
$\F_q^d\setminus\{a\}$, and since each has $q-1$ points, their number is
$H=(N-1)/(q-1)$.

The geometric-series formula gives $H\geq q^{d-1}=N/q$ and
$H\leq q^{d-1}q/(q-1)\leq2q^{d-1}$.  Finally,
\[
  \frac{Hq}{N}
  =\frac{q}{q-1}\left(1-\frac1N\right)
  <\frac{q}{q-1}.
\]
\end{proof}

Next we show that degree $d=\log_2R$ suffices to obtain enough independence for our argument.
Let $\cP$ be the $\F_q$-vector space of $d$-variate polynomials of total
degree at most $d$.

\begin{lemma}[Independent evaluations]
\label{lem:independent-evaluations}
There are distinct anchors $a_1,\ldots,a_R\in\F_q^d$ such that, for a
uniformly random $P\in\cP$, the values
\[
  b_t:=P(a_t),\qquad t\in[R],
\]
are jointly uniform on $\F_q^R$.  Consequently, for every $t$, the value
$b_t$ is uniform in $\F_q$ conditional on $b_1,\ldots,b_{t-1}$.
\end{lemma}
A more general version of this statement is also known from coding theory
(as the information-set theorem for generalized Reed--Muller codes~\cite[Theorem~1]{KeyMcDonoughMavron2006});
we give a simpler self-contained proof.

\begin{proof}
The polynomial space has dimension
\[
  \dim\cP=\binom{2d}{d}\geq2^d=R.
\]
We define the evaluation-on-$\F_q^d$ map on $\cP$
by $E(P) := (P(a))_{a \in \F_q^d}$.
It is well-known that $E : \cP \to \F_q^{\F_q^d}$ is injective,
i.e., that $E(Q)$ being zero implies $Q = 0$.
We sketch a proof for completeness.
Induct on the number $r$ of variables.
The case $r=0$ is trivial.
For $r \ge 1$, suppose that
$
  Q(X_1,\ldots,X_r)
  =\sum_{i=0}^{d}Q_i(X_1,\ldots,X_{r-1})X_r^i
$
vanishes on $\F_q^r$.  After fixing the first $r-1$ coordinates, the resulting
univariate polynomial (in variable $X_r$) has degree at most $d<q$ and vanishes on all of
$\F_q$, so it is zero (a nonzero univariate polynomial would have at most $d$ roots); thus every coefficient $Q_i$ vanishes on $\F_q^{r-1}$.  The induction
hypothesis gives $Q_i=0$ for every $i$, and hence $Q=0$.

Since $E$ is injective, its rank is $\dim(\cP) \ge R$.
As $E$ consists of evaluation functionals $P \mapsto P(a)$ for $a \in \F_q^d$,
we can choose $R$ linearly independent ones, indexed by points
$a_1,\ldots,a_R$.\footnote{In fact, one can show that the vertices of the $d$-dimensional hypercube form a valid choice.}  The map
\[
  P\longmapsto(P(a_1),\ldots,P(a_R))
\]
then has rank $R$ and is surjective onto $\F_q^R$.  A uniform input to a
surjective linear map has uniform output, proving joint uniformity.
\end{proof}

Finally, we give
a slightly weaker consequence of the Johnson bound for list
recovery~\cite[Lemma~5.2]{GopiEtAl2018}, specialized to univariate
polynomials evaluated on $\F_q$.  We include the short pairwise-agreement
proof for completeness.

\begin{lemma}[Johnson-type univariate list recovery]
\label{lem:johnson}
For each $x\in\F_q$, let
$A_x\subseteq\F_q$ have size at most $s$.  Call a univariate polynomial $p$
of degree at most $d$ \emph{$\rho$-prolific} if
\[
  |\{x\in\F_q:p(x)\in A_x\}|\geq\rho q.
\]
Then the set $\mathcal C$ of $\rho$-prolific polynomials
satisfies
\[
  |\mathcal C| \cdot (\rho^2-sd/q) \leq s.
\]
\end{lemma}

\begin{proof}
Define
\[
  m_{x,z}=|\{p\in\mathcal C:p(x)=z\}|.
\]
For a prolific $p$, the number of $x \in \F_q$ with $z=p(x) \in A_x$ is at least $\rho q$, so
$\sum_{x\in\F_q}\sum_{z\in A_x}m_{x,z} \geq \rho q |\mathcal C|$.
Applying Cauchy--Schwarz to the at most $qs$ pairs
$(x,z)$ with $z\in A_x$ gives
\begin{equation}
  |\mathcal C|^2\rho^2q^2
  \leq\left(\sum_{x\in\F_q}\sum_{z\in A_x}m_{x,z}\right)^2
  \leq qs\sum_{x\in\F_q}\sum_{z\in A_x}m_{x,z}^2
  \leq qs\sum_{x,z \in \F_q}m_{x,z}^2.
  \label{eq:johnson-cs}
\end{equation}
The final sum counts ordered triples $(p,p',x)$ with $p(x)=p'(x)$.  Equal
polynomial pairs contribute $|\mathcal C|q$.  Two distinct degree-at-most-$d$
polynomials agree at no more than $d$ points, so
\[
  \sum_{x,z \in \F_q}m_{x,z}^2
  \leq |\mathcal C|q+|\mathcal C|(|\mathcal C|-1)d
  \leq |\mathcal C|q+|\mathcal C|^2d.
\]
Substitute this into~\eqref{eq:johnson-cs} and divide by
$|\mathcal C|q^2$ to obtain
the claim.
(The claim is trivial if $|\mathcal C| = 0$.)
\end{proof}

\section{The hidden-polynomial instance}
\label{sec:hidden-polynomial}

In this section, we construct a Min-Rep instance along with a hard distribution of online constraint sequences. The Min-Rep instance (i.e., the variables, labels, possible constraints and their relations) is fixed, while the online constraint sequence is determined by a hidden random polynomial.

Recall parameters $R,d,q,N$ from the beginning of \Cref{sec:preliminaries}: $R$ is a sufficiently large power of two, $d=\log_2R$, $q=R^{10}$, and $N=q^d$.

\paragraph{Anchors.}
Fix anchors $a_1,\ldots,a_R\in \F^{d}_{q}$ satisfying
\cref{lem:independent-evaluations}.

\paragraph{Variables and Labels.} The right-variable set is $V=\F_q^d$.  Each $x\in V$ has label set
$\Lambda_x=\F_q$, with cost 
\[
c_p = 100R
\]
per label. Each right variable corresponds to a point in $\F^{d}_{q}$, and we refer to the labels of the right variables as \emph{point labels}.

The left-variable set is 
\[
  U=\{(t,b,\ell):t\in[R],\ b\in\F_q,
       \ \ell\text{ is an affine line through }a_t\}.
\]
We now explain the intuition behind the tuple $(t,b,\ell)$, which will become clearer when we define the online constraint sequence. Here, $t$ denotes the \emph{round} number, $b$ corresponds to the possible values of the hidden polynomial at the anchors, and $\ell$ is an affine line through the anchor $a_{t}$ of round $t$.

We introduce the notion of \emph{line polynomials} before defining the labels of left variables.

\begin{definition}[Line polynomials]
Consider an affine line $\ell$ containing the anchor point $a_t$. 
We fix an affine parametrization $\phi_{t,\ell}: \F_q \to \ell$ defined by $\phi_{t,\ell}(s) = a_t + s \cdot v$, where $v \in \F_q^d \setminus \{0\}$ denotes an arbitrary direction vector of $\ell$. Note that this parametrization is bijective.

For any univariate polynomial $\tilde{g}: \F_q \to \F_q$, we refer to the composed function $g = \tilde{g} \circ \phi_{t,\ell}^{-1}: \ell \to \F_q$ as a \emph{line polynomial} on $\ell$. We call $\tilde{g}$ the \emph{underlying polynomial} of $g$.
The line polynomial $g$ has degree at most $d$ if its underlying polynomial $\tilde{g}$ has degree at most $d$.
\end{definition}

For each left-variable $(t,b,\ell)\in U$, its label set $\Lambda_{t,b,\ell}$ consists of the
degree-at-most-$d$ line polynomials $g$ on $\ell$ satisfying $g(a_t)=b$, and each label has cost
\[c_r = 100q.\]
Similarly, we refer to the labels of the left variables as \emph{line labels}.

\paragraph{Constraints and Relations.}

For every left-variable $(t,b,\ell)\in U$ and $x\in\ell$, we add a constraint between it and the right-variable $x\in \ell\subseteq V$ on this line $\ell$. The relation of this constraint is
\[
  \mathcal R_{(t,b,\ell),x}
  :=\{(g,z):g\in\Lambda_{t,b,\ell},\ z\in\Lambda_x,\ g(x)=z\}.
\]
In other words, for each line polynomial $g$ in $\Lambda_{t,b,\ell}$, we add one pair $(g,z = g(x))$ to the relation.

\paragraph{The Online Constraint Sequence.}
For $t\in[R]$ and $b\in\F_q$, the \emph{block} $(t,b)$ consists of all constraints whose left variables are associated with $(t,b)$. 

Choose the hidden polynomial $P:\F^{d}_{q}\to \F_{q}$ uniformly from $\cP$
(recall from \Cref{lem:independent-evaluations} that $\cP$ collects all $d$-variate polynomials of total degree at most $d$). For each round $t$, let $b_t=P(a_t)$ be the evaluated value of $P$ at the anchor $a_{t}$. 

The online constraint sequence $\sigma(P)$ reveals constraints round by round from $t=1$ to $R$. In round $t$, all constraints in the block $(t, b_{t})$ are revealed. Namely,
\[
  \sigma(P):=((1,b_1),(2,b_2),\ldots,(R,b_R)).
\]
Constraints within each block appear in an arbitrary order fixed in advance.

\paragraph{Properties of the Instance.} The following \Cref{lem:hidden-instance-properties} bounds the number of labels and constraints in the Min-Rep instance. The following \Cref{lem:offline} upper bounds the offline optimum of any online constraint sequence in the distribution.

\begin{lemma}
\label{lem:hidden-instance-properties}
The above gives a valid weighted online Min-Rep
instance with
$Nq$ right labels,
$RqHq^d$ left labels,
and $RHq^2$ possible constraints.
Every block has $Hq$ constraints, and every sequence $\sigma(P)$ has $RHq$
constraints.  
\end{lemma}

\begin{proof}
There are $N$ right variables with $q$ labels each.  There are $RqH$ left
variables.  A degree-at-most-$d$ univariate polynomial $\tilde{g}$ with value
$\tilde{g}(0)=b$ at zero has $d$ free coefficients, and hence $q^d$ choices.  Because
$d<q$, distinct such polynomials induce distinct functions on $\F_q$. Moreover, such polynomials as underlying polynomials lead to $q^{d}$ distinct line polynomials $g$ with degree at most $d$ and $g(a_{t}) = b$, so each
left variable has exactly $q^d$ labels.

Every displayed relation is nonempty: for example, the constant polynomial
$g=b$ is a left label and is compatible at $x$ with $z=b$.  Each ordered
variable pair occurs at most once because $(t,b,\ell)$ and $x$ uniquely
determine the relation.  Each left variable is constrained to the $q$ points
of its line.  Hence there are $RqHq=RHq^2$ possible constraints, while fixing
$(t,b)$ leaves $Hq$ constraints in one block and $RHq$ in a generated
sequence.  
\end{proof}

\begin{lemma}[Offline optimum]
\label{lem:offline}
For every $P\in\cP$, the sequence $\sigma(P)$ satisfies
\[
  \OPT_{\mathrm{MinRep}}(\sigma(P))<250RN< M:=300RN.
\]
\end{lemma}

\begin{proof}
The offline solution is as follows. 
\begin{itemize}
\item For each right variable $x\in \F^{d}_{q}$, buy label $P(x)\in \Lambda_x$.
\item For each round $t$ and each line $\ell$ through $a_{t}$ (corresponding to the left variable $(t,b_{t},\ell)$), buy label $P|_\ell \in \Lambda_{t,b_{t},\ell}$. Indeed, $P|_\ell$ is a degree-at-most-$d$ line polynomial on $\ell$ (furthermore $P(a_{t})=b_{t}$ so it is in $\Lambda_{t,b_{t},\ell}$), since $P$ has total degree at most $d$, and the parametrization $\phi_{t,\ell}$ is affine-linear.
\end{itemize}
We can easily verify that this offline solution satisfies all the constraints in $\sigma(P)$. The label cost is
\[
  100RN+100qRH,
\]
since we buy $N$ point labels and $RH$ line labels (recall that $R$ is the number of rounds and $H$ is the number of affine lines through $a_{t}$).
By \cref{lem:pencil-geometry}, we have $Hq/N<q/(q-1)\leq4/3$ for $q\geq4$.
Hence this cost is less than
\[
  100RN+100\cdot\frac43RN<250RN<M.
\]
\end{proof}

\section{The lower bound}
\label{sec:distributional-lower-bound}

We consider deterministic algorithms for now.
Since we give a hard distribution of constraint sequences
over a fixed Min-Rep instance,
it will be easy to generalize to randomized algorithms (against an oblivious adversary)
later in the proof of \cref{thm:main}.

Fix a deterministic weighted online Min-Rep strategy. Our goal is to show that its expected cost over the hard distribution of constraint sequences is significantly larger than the offline cost upper bound $M$ (i.e., \Cref{prop:distributional}). Recall that each constraint sequence in the hard distribution is revealed round by round. For simplicity, we assume all constraints within a round are revealed simultaneously at the start of that round. Note that this only makes the online task easier. Furthermore, without loss of generality, assume that labels of a left variable $(t,b_t,\ell)$ are only bought during round $t$.

We first introduce some notation. Consider a round $t$.
\begin{itemize}
\item Let $C_{t}$ denote the online cost up to the end of round $t$. In particular, at the very beginning, the online cost is $C_{0} = 0$. Write $\Delta C_{t} = C_{t}-C_{t-1}$.
\item For each point $x\in \F_{q}^{d}$ (i.e., each right variable), let $L^{<t}_{x}\subseteq \Lambda_{x}$ be the point labels bought before round $t$.
\item For each line $\ell$ through $a_{t}$ (i.e., each left variable $(t, b_{t},\ell)$), let $G_{t,\ell}\subseteq \Lambda_{t,b_{t},\ell}$ be the line labels bought in round $t$.
\end{itemize}

The argument now has four steps, as explained in \cref{sec:online-lb-tech-overview}.  First, old point labels determine a short
list of possible anchor values on each line.  Second, below the target cost,
most lines contain many points with short old lists.  Third, whenever the new
anchor value misses a line's set of prolific polynomials, that line forces a large cost.
Finally, we sum these conditional round costs until the target is reached.

\begin{definition}[Heavy points and prolific polynomials]
\label{def:prolific}
Consider a round $t$. A point $x$ is \emph{$t$-heavy} if $|L_x^{<t}|>16R$.  Define its truncated
old-label list by
\[
  \widehat L_{x,t}=
  \begin{cases}
    L_x^{<t},&|L_x^{<t}|\leq16R,\\
    \varnothing,&|L_x^{<t}|>16R.
  \end{cases}
\]
Define 
\[
\rho=1/(256R).
\]
For a line $\ell$ through $a_t$, a degree-at-most-$d$
line polynomial $g$ on $\ell$ is \emph{$(t,\ell)$-prolific} if
\[
  |\{x\in\ell:g(x)\in\widehat L_{x,t}\}|\geq\rho q.
\]
Let $\mathcal C_{t,\ell}$ be the set of these prolific polynomials, and let
$D_{t,\ell}$ be the event that 
\[
b_t\in\{g(a_t):g\in\mathcal C_{t,\ell}\}
\]
holds. Namely, $D_{t,\ell}$ occurs if some line polynomial from $\mathcal{C}_{t,\ell}$ has value $b_{t}$ at the anchor $a_{t}$.
\end{definition}

We note that all objects in the above \Cref{def:prolific}, except for whether $D_{t,\ell}$ occurs, are
determined before the beginning of round $t$. For better understanding, we point out that such a line polynomial $g$ being $(t,\ell)$-prolific means, by definition, that its underlying polynomial $\tilde{g}=g\circ \phi_{t,\ell}$ is prolific (w.r.t. $A_{\tilde{x}}=\widehat{L}_{x,t}$ for each $\tilde{x}\in \F_{q}$ and $x = \phi_{t,\ell}(\tilde{x})$) under the definition in \Cref{lem:johnson}. 

The following \Cref{lem:prolific-list} bounds the number of prolific line polynomials, and further bounds the probability that $D_{t,\ell}$ occurs.

\begin{lemma}[The prolific list rarely predicts the anchor]
\label{lem:prolific-list}
For each round $t$ and every line $\ell$ through $a_t$,
\[
  |\mathcal C_{t,\ell}|
  \leq 2^{21}R^3
\]
and
\begin{equation}
  \Prb[D_{t,\ell}\mid b_1,\ldots,b_{t-1}]\leq\frac14.
  \label{eq:prolific-hit-probability}
\end{equation}
\end{lemma}

\begin{proof}
For sufficiently large $R$, the parameter choices give
$16Rd/q\leq\rho^2/2$. 

Apply \cref{lem:johnson} with respect to $A_{\tilde{x}}=\widehat{L}_{x,t}$ for each $\tilde{x}\in \F_{q}$ and $x = \phi_{t,\ell}(\tilde{x})$ (so $s=16R$). It gives that the number of prolific degree-at-most-$d$ underlying polynomials is at most $s/(\rho^{2}-sd/q)$. Hence, the size of $\mathcal{C}_{t,\ell}$, i.e., the number of $(t,\ell)$-prolific degree-at-most-$d$ line polynomials is at most
\[
  |\mathcal C_{t,\ell}|
  \leq\frac{16R}{\rho^2-16Rd/q}
  \leq\frac{32R}{\rho^2}
  =2^{21}R^3.
\]
Conditional on $b_1,\ldots,b_{t-1}$, the set
$\{g(a_t):g\in\mathcal C_{t,\ell}\}$ is fixed, while $b_t$ is uniform in
$\F_q$ by \cref{lem:independent-evaluations} and the way we sample the hidden polynomial $P$. Hence
\[
  \Prb[D_{t,\ell}\mid b_1,\ldots,b_{t-1}]
  \leq\frac{2^{21}R^3}{q}\leq\frac14.
\]
\end{proof}

Recall that for a line $\ell$ through the anchor $a_{t}$, $\ell$'s nonanchor point set is $\ell\setminus \{a_{t}\}$. \Cref{def:crowded} defines crowded lines, which are lines with many heavy nonanchor points. \Cref{lem:many-noncrowded} shows that, if the cost spent before round $t$ is low, then in round $t$, a large fraction of lines through $a_{t}$ is not crowded.

\begin{definition}[Crowded lines]
\label{def:crowded}
In round $t$, a line $\ell$ through $a_t$ is \emph{$t$-crowded} if
$\ell\setminus\{a_t\}$ contains at least $q/2$ $t$-heavy points.
\end{definition}

\begin{lemma}[Many lines are not crowded]
\label{lem:many-noncrowded}
Consider a round $t$. If $C_{t-1}<RM$, then at least $5H/8$ lines through $a_t$ are
not $t$-crowded, and for each such line $\ell$, at least $q/2$ of its nonanchor
points are not $t$-heavy.
\end{lemma}

\begin{proof}
Let $h_t$ be the number of $t$-heavy points in $\F^{d}_{q}$. Since every point label costs
$c_p$ and every heavy point has more than $16R$ old labels bought in previous rounds, we have
\[
  16Rh_t
  \leq\sum_{x\in\F_q^d}|L_x^{<t}|
  \leq\frac{C_{t-1}}{c_p}
  <\frac{RM}{c_p}=3RN.
\]
Hence $h_t<3N/16$.  Distinct lines through $a_t$ have disjoint nonanchor
points, so
fewer than
\[
  \frac{3N/16}{q/2}=\frac{3N}{8q}\leq\frac{3H}{8}
\]
lines are $t$-crowded, where the inequality uses
\cref{lem:pencil-geometry}.  

The second statement is essentially by definition. Every non-crowded line has $q-1$ nonanchor points and at most
$q/2-1$ of them are $t$-heavy. It therefore contains at least $q/2$ nonheavy
nonanchor points.
\end{proof}

\begin{lemma}[Per-line label-cost tradeoff]
\label{lem:line-tradeoff}
Consider a round $t$ with $C_{t-1}<RM$, and a line $\ell$ through $a_t$ that is not
$t$-crowded.  If $D_{t,\ell}$ does not occur, then in this round, a total cost of at least 
\[
50Rq
\]
must be spent on buying new line labels of the left variable $(t,b_{t},\ell)$ and new point labels of $\ell$'s nonanchor points.
\end{lemma}

\begin{proof}
Let $u = |G_{t,\ell}|$ be the number of line labels of $(t,b_t,\ell)$ bought during round $t$ (recall that, by assumption, no line labels for $(t,b_t,\ell)$ were bought in previous rounds).  

When $D_{t,\ell}$ does not occur, every line polynomial
$g\in G_{t,\ell}$ is not $(t,\ell)$-prolific: every $g\in G_{t,\ell}$ has $g(a_t)=b_t$, whereas no $(t,\ell)$-prolific line polynomial has value $b_{t}$ at anchor $a_{t}$.

Since $\ell$ is not $t$-crowded, it contains at least $q/2$ nonheavy
nonanchor points. At each such point $x$, the revealed constraint between $(t,b_{t},\ell)$ and $x$ requires that the point label $g(x)$ of $x$ be bought for some $g\in G_{t,\ell}$.
If this point label was an old label bought before round $t$, then
$g(x)\in\widehat L_{x,t}$.  Each of the $u$ nonprolific polynomials can
support old labels at fewer than $\rho q$ points.  Hence at least
$
  q/2-\rho q u
$
distinct nonanchor point labels must be bought during round $t$. The $u$ line labels cost $100qu$, while each new point label costs $100R$.
Their total cost is therefore at least
\[
  100qu+100R(q/2-\rho q u)
  =100q\left(\frac R2+(1-R\rho)u\right)
  =100q\left(\frac R2+\frac{255}{256}u\right)
  \geq50Rq.
\]
\end{proof}

\begin{lemma}[Expected cost of an unfinished round]
\label{lem:round-increment}
For each round $t$,
for every realization $(b_1, \ldots, b_{t-1})$ with $C_{t-1}<RM$,
\begin{equation}
  \E[\Delta C_t\mid b_1,\ldots,b_{t-1}]\geq\frac{5M}{64}.
  \label{eq:round-increment}
\end{equation}
\end{lemma}

\begin{proof}
For distinct lines through $a_t$, the left variables are distinct and the
nonanchor point sets are disjoint.  \Cref{lem:line-tradeoff} therefore
gives
\[
  \Delta C_t
  \geq
  50Rq
  \sum_{\substack{\text{line $\ell$ through $a_{t}$}\\
                  \ell\text{ not }t\text{-crowded}}}
  (1-\mathbf{1}_{D_{t,\ell}}).
\]
By \cref{lem:prolific-list,lem:many-noncrowded}, and by
linearity of conditional expectation,
\[
  \begin{aligned}
  \E[\Delta C_t\mid b_1,\ldots,b_{t-1}]
  &\geq\frac{5H}{8}\cdot\frac34\cdot50Rq
   =\frac5{64}(300RHq)\\
  &\geq\frac5{64}(300RN)
   =\frac{5M}{64},
  \end{aligned}
\]
where $Hq\geq N$ by \cref{lem:pencil-geometry}.
\end{proof}

\begin{proposition}[Deterministic distributional lower bound]
\label{prop:distributional}
Every deterministic weighted online Min-Rep strategy
satisfies
\begin{equation}
  \E_P[C_R]\geq\frac{5R}{128}M.
  \label{eq:distributional}
\end{equation}
\end{proposition}

\begin{proof}
If $\Prb[C_R\geq RM]\geq1/2$, then $\E[C_R]\geq RM/2$, which is stronger than
\eqref{eq:distributional}.  Otherwise, monotonicity of the accumulated cost
gives
\[
  \Prb[C_{t-1}<RM]\geq\Prb[C_R<RM]>\frac12
\]
for every $t$.  Since $\Delta C_t\geq0$, the tower property and
\eqref{eq:round-increment} give
\[
  \E[\Delta C_t]
  \geq
  \E\!\left[
    \mathbf 1_{\{C_{t-1}<RM\}}
    \E[\Delta C_t\mid b_1,\ldots,b_{t-1}]
  \right]
  \geq\frac{5M}{64}\Prb[C_{t-1}<RM]
  >\frac{5M}{128}.
\]
Summing over the $R$ rounds proves~\eqref{eq:distributional}.
\end{proof}

Finally we are ready to give the proof of our main theorem.
We restate it for convenience.

\mainthm*
\begin{proof}
Apply \cref{lem:realization} to the Min-Rep instance, and call the resulting unit-cost DAG
$J_R$.  By \cref{lem:hidden-instance-properties}, this graph is
independent of $P$.  Every sequence $\sigma(P)$ has $RHq$ constraints, so
\cref{lem:offline} and the lifting direction of
\cref{lem:realization} give
\begin{equation}
  \OPT_{J_R}(\sigma(P))
  \leq \OPT_{\mathrm{MinRep}}(\sigma(P))+3RHq
  <250RN+4RN<M,
  \label{eq:dsn-offline}
\end{equation}
where $3RHq<4RN$ follows from \cref{lem:pencil-geometry}.

Now let $\mathcal A$ be any randomized online \DSN{} algorithm on $J_R$ and
fix its random tape $\omega$.  \Cref{lem:realization} turns the
resulting deterministic \DSN{} algorithm into a deterministic weighted online
Min-Rep strategy whose final label cost, denoted $C_R(P,\omega)$, is at most
$\cost_{\mathcal A}(\sigma(P),\omega)$.  By applying \cref{prop:distributional} we have,
 for every $\omega$,
\[
  \E_P[\cost_{\mathcal A}(\sigma(P),\omega)]
  \geq \E_P[C_R(P,\omega)]
  \geq\frac{5R}{128}M.
\]
Average over the random tape.  Since $P$ ranges over the finite set $\cP$ and
all costs are nonnegative, interchanging the two expectations gives
\[
  \E_P\E_\omega[\cost_{\mathcal A}(\sigma(P),\omega)]
  \geq\frac{5R}{128}M.
\]
Since $P$ is uniform on the finite set $\cP$, at least one fixed polynomial
$P^*\in\cP$ satisfies
\begin{equation}
  \E_\omega[\cost_{\mathcal A}(\sigma(P^*),\omega)]
  \geq\frac{5R}{128}M
  \geq\frac{5R}{128}
        \OPT_{J_R}(\sigma(P^*)).
  \label{eq:oblivious-sequence}
\end{equation}
The second inequality uses~\eqref{eq:dsn-offline}.  The sequence $\sigma(P^*)$
is fixed before the algorithm draws its random bits, so the adversary is
oblivious.

We now express the gap in terms of the number of vertices.  By
\cref{lem:hidden-instance-properties}, there are $Nq$ point-label chains,
$RqHq^d$ line-label chains, and $RHq^2$ possible constraints.  Counting the
vertices on all label chains and the two terminals for every possible
constraint gives
\begin{equation}
  n_R
  \leq(c_p+1)Nq+(c_r+1)RqHq^d+2RHq^2.
  \label{eq:vertex-count}
\end{equation}
Using $H\leq2q^{d-1}$ from \cref{lem:pencil-geometry}, together with
$c_p=100R$ and $c_r=100q$, this is
\[
  n_R=O(Rq^{2d+1}).
\]
Since $q=R^{10}$ and $d=\log_2R$,
\[
  \log n_R=O((\log R)^2).
\]
Consequently, for some absolute $c>0$ and all sufficiently large $R$,
\[
  \frac{5R}{128}\geq\exp\!\bigl(c\sqrt{\log n_R}\bigr).
\]
Together with~\eqref{eq:oblivious-sequence}, this proves the claimed lower
bound.
\end{proof}

\section*{AI Disclosure}
The proof was initially discovered using a custom agentic harness developed by the authors,
which used GPT 5.6 Sol, Opus 4.8, and Gemini 3.1 Pro.
The human authors have verified, simplified, and rewritten the proof.
AI was also used in drafting and polishing.
The authors assume all responsibility for the paper's content and correctness.

\bibliographystyle{alpha}
\bibliography{references.bib}

\appendix

\section{Extension of Dodis-Khanna to unit costs}
\label{app:offline-unitization}

Dodis and Khanna's lower bound for \DSN{} in \Cref{thm:dodis-khanna} is stated for weighted graphs~\cite{DodisKhanna1999}. In this appendix, we give a simple reduction which extends their lower bound to unit-cost graphs. This appendix is independent of our online  lower bound construction.

We begin with a reduction from weighted \DSN{}  to the unit cost setting. In our  \DSN{} instances, let $n$ be the number of vertices, $m \le n^2$ the number of edges, and $k \le n^2$ the number of distinct demands. We assume that the weighted \DSN{} instances have polynomial aspect ratio (i.e., the ratio of the largest weight divided by the smallest weight is polynomially bounded).

\begin{proposition}[Weighted-to-unit-cost transfer]
\label{prop:offline-unitization}
Suppose unit-cost \DSN{} has a polynomial-time $\rho(N)$-approximation on
$N$-vertex graphs, where $\rho$ is nondecreasing.  Then weighted \DSN{} has a
polynomial-time $4\rho(n+n^4)$-approximation.  If the unit-cost guarantee is
instead expressed as $\rho(k)$ in the number of demands, the weighted
guarantee is $4\rho(k)$.
\end{proposition}

\begin{proof}
Let $c_{\min}$ be the smallest positive edge cost and let
$C=\sum_{e}c_e$.  Then 
\[
  c_{\min}\leq\OPT\leq C.
\]
We can guess the value of $\OPT$ within a factor of two by making $O(\log(C/c_{\min}))$ guesses of the form $B=2^i c_{\min}$. One of these guesses satisfies
\begin{equation}
  \OPT\leq B<2\OPT.
  \label{eq:unitization-correct-guess}
\end{equation}
For each guess, discard edges of cost greater than $B$.  Replace every retained
edge $e=(u,v)$ by a directed $(u, v)$-path $P_e$ with $|P_e| = L_e$ edges, where 
\[
  L_e=\max\left\{1,\left\lceil\frac{mc_e}{B}\right\rceil\right\}.
\]
Since $c_e\leq B$, every $L_e\leq m$, and the subdivided instance has at most
$n+m^2$ vertices.  Subdivision also preserves acyclicity whenever the original
graph is acyclic.

For a guess satisfying~\eqref{eq:unitization-correct-guess}, every edge of an
optimal weighted solution is retained.  Replacing these edges by their paths
gives a feasible unit-cost solution of value at most
\[
  \sum_{e\in F^*}L_e
  \leq |F^*|+\frac{m}{B}\sum_{e\in F^*}c_e
  \leq m+m=2m.
\]
Conversely, let $H$ be any feasible edge set in the subdivided graph, and let
\[
  F(H)=\{e:\text{the entire path replacing $e$ belongs to $H$}\}.
\]
Every path between original vertices that enters an internal subdivision
vertex must traverse the corresponding replacement path completely.  Thus
$F(H)$ satisfies all original demands.  Replacement paths are edge-disjoint
and $L_e\geq mc_e/B$, so
\begin{equation}
  c(F(H))\leq\frac{B}{m}|H|.
  \label{eq:unitization-projection}
\end{equation}

Run the hypothetical unit-cost algorithm for every feasible guess, project
its output by~\eqref{eq:unitization-projection}, and return the cheapest
projection.  On the correct guess its cost is less than
\[
  \frac{B}{m}\,\rho(n+m^2)\,(2m)
  <4\rho(n+n^4)\OPT.
\]
This reduction does not change the demand set, so we obtain the $4\rho(k)$ approximation as well.
\end{proof}

We can now complete the extension to unit costs. 

\begin{corollary}[Offline unit-cost hardness]
\label{cor:offline-unit-hardness}
For every constant $\varepsilon>0$, unit-cost \DSN{} admits no polynomial-time
\[
  2^{(\log N)^{1-\varepsilon}}
\]
approximation on $N$-vertex graphs unless
$\mathsf{NP}$ has quasipolynomial-time algorithms.
The same conclusion holds for an approximation guarantee
$2^{(\log k)^{1-\varepsilon}}$ expressed in the number of demands.
\end{corollary}

\begin{proof}
Suppose the stated unit-cost approximation existed.  In the
$N$-dependent case put $X=n+n^4$, and in the $k$-dependent case put $X=k$.
By \cref{prop:offline-unitization}, weighted \DSN{} would in either case have
approximation ratio
\[
  4\cdot2^{(\log X)^{1-\varepsilon}}.
\]
We have $\log X=O(\log n)$: this is immediate in the first case, while in
the second case the initial simplification gives $k\leq n^2$.  Hence, for
every fixed $0<\delta<\varepsilon$, the displayed ratio is at most
$2^{(\log n)^{1-\delta}}$ for all sufficiently large $n$, contradicting the
theorem of Dodis and Khanna~\cite{DodisKhanna1999} unless
$\mathsf{NP}$ has quasipolynomial-time algorithms.
\end{proof}

\section{Proof for submodular cost covering}
\label{app:other-consequences}

\begin{proof}[Proof of \Cref{cor:submodular-cost-covering}]
Start with the weighted online Min-Rep instance from
\cref{sec:hidden-polynomial}.  Its resources are the variable--label pairs
\[
  \mathcal Q:=\{(w,\lambda):w\in U\cup V,\ \lambda\in\Lambda_w\},
\]
where resource $(w,\lambda)$ has weight $c(w,\lambda)$.  For every possible
constraint $e=(u,v)$ and every compatible pair
$(\lambda,\mu)\in\mathcal R_{uv}$, create an action
$a_{e,\lambda,\mu}$.  This action activates the two resources
$(u,\lambda)$ and $(v,\mu)$.  Write $\mathcal A_e$ for the actions tagged by
$e$; these sets are pairwise disjoint.  Introduce a variable
$y_a\in\{0,1\}$ for every action $a$.  When constraint $e$ arrives, reveal
the upward-closed covering constraint
\[
  \max_{a\in\mathcal A_e}y_a\geq1.
\]

For $y\in\mathbb R_+^{\mathcal A}$, define
\[
  f(y):=
  \sum_{r\in\mathcal Q}c(r)
  \max\bigl(\{y_a:r\in\mathcal Q(a)\}\cup\{0\}\bigr),
\]
where $\mathcal Q(a)$ is the two-resource set activated by $a$.  We note that $f$ is continuous,
monotone, and submodular over $\mathbb R_+^{\mathcal A}$.  On the Boolean domain it
is exactly the weighted coverage function that charges once for every
resource activated by at least one selected action.  Thus these variables,
domains, constraints, and the objective form an instance of the framework of
Koufogiannakis and Young.

Selecting actions online induces a Min-Rep strategy that purchases every
activated resource.  The two strategies have the same cost at every prefix,
and the covering requirement for $e$ supplies a compatible label pair
satisfying $e$.  Conversely, from any offline Min-Rep solution, choose one
compatible pair for each requested constraint and select the corresponding
action.  Its cost is at most the label cost.  Therefore the distributional
gap of \cref{prop:distributional} carries over unchanged.

It remains to express the gap in terms of the number of actions.  The hidden
instance has $RHq^2$ possible constraints.  Each relation contains exactly
$q^d$ compatible pairs: every one of the $q^d$ line labels determines a
unique compatible point label.  Thus
\[
  m=RHq^{d+2}\leq 2Rq^{2d+1},
\]
using $H\leq2q^{d-1}$.  Since $q=R^{10}$ and $d=\log_2R$, we have
$\log m=O((\log R)^2)$.  The $\Omega(R)$ competitive gap is consequently
$\exp(\Omega(\sqrt{\log m}))$.
\end{proof}

\end{document}